\documentclass[prl, amsfonts, amssymb, amsmath, nofootinbib,reprint, showkeys, twoside,notitlepage]{revtex4-2}
\def\prc{{Phys.~Rev.~C}}        % Physical Review C
\def\prd{{Phys.~Rev.~D}}        % Physical Review D
\def\prl{{Phys.~Rev.~Lett.}}    % Physical Review Letters
\usepackage{amsmath,amstext}
\usepackage[T1]{fontenc}
\usepackage{amssymb}
\input{epsf}
\usepackage{graphicx}
\usepackage{ae,aecompl}

\usepackage{hyperref}
\usepackage{amsmath}
\usepackage{amssymb}
\usepackage{mathtools}
\usepackage{bm}
\usepackage{cleveref}
\usepackage{tensor}
\usepackage{braket}
\usepackage{enumitem}
\usepackage{mhchem}
\usepackage{cancel}
\usepackage{amsthm}
\usepackage{upgreek}
\usepackage{mathrsfs}

\usepackage{color}

\def\be{\begin{equation}}
\def\ee{\end{equation}}
\def\beq{\begin{eqnarray}}
\def\eeq{\end{eqnarray}}

\theoremstyle{definition}

\theoremstyle{theorem}
\newtheorem{theorem}{Theorem}

\theoremstyle{corollary}

\begin{document}
\title{Dispersion relations alone cannot guarantee causality}
\author{L.~Gavassino$^1$, M.~Disconzi$^{1}$, \& J.~Noronha$^2$}
\affiliation{
$^1$Department of Mathematics, Vanderbilt University, Nashville, TN, USA
\\
$^2$Illinois Center for Advanced Studies of the Universe \& Department of Physics,
University of Illinois at Urbana-Champaign, Urbana, IL 61801-3003, USA
}

\begin{abstract}
We show that linear superpositions of plane waves involving a single-valued, covariantly stable dispersion relation $\omega(k)$ always propagate outside the lightcone unless $\omega(k) =a+b k$. This implies that there is no notion of causality for individual dispersion relations since no mathematical condition on the function $\omega(k)$ (such as the front velocity or the asymptotic group velocity conditions) can serve as a sufficient condition for subluminal propagation in dispersive media. Instead, causality can only emerge from a careful cancellation that occurs when one superimposes all the excitation branches of a physical model. This happens automatically in local theories of matter that are covariantly stable. Hence, we find that the need for non-hydrodynamic modes in relativistic fluid mechanics is analogous to the need for anti-particles in relativistic quantum mechanics.
\end{abstract}

\maketitle

\section{Introduction}
The ``practical'' definition of relativistic causality is universally accepted: it is impossible to transmit information faster than the vacuum speed of light \cite{Einstein1908,Tolman_book,landau2}. The question is how to translate this principle into mathematical constraints that our physical theories must obey. In some cases, this question has an unambiguous answer. In classical field theory, causality demands that the characteristics of the field equations lay inside or upon the lightcone \cite{CourantHilbert2_book,Susskind1969,Hawking1973,Wald,Rauch_book}. In quantum field theory, the commutator of spacelike-separated observables must vanish \cite{Peskin_book,Eberhard1988,Keister1996}. 
In other contexts, the mathematical nature of causality is less understood.

Consider a homogeneous system in thermodynamic equilibrium, and let $\omega(k)$ be the eigenfrequency of one of its (linear) excitations, which is a function\footnote{%The system is assumed to be homogeneous, but it does not need to be isotropic. 
For anisotropic systems, the dispersion relation is $\omega(k^x,k^y,k^z)$. In this case, we align the $x-$axis along a direction of interest and define $\omega(k):=\omega(k,0,0)$. Then, our results apply to waves propagating in $x$.} of the wavenumber $k$.
Under which conditions is such dispersion relation compatible with causality? 
Most attempted answers revolve around imposing inequalities on the phase velocity $(\mathfrak{Re}\omega)/k$, or on the group velocity $d(\mathfrak{Re}\omega)/dk$ \cite{Brillouin_book,DenicolKodamaMota2008,Koide2011}. However, no fully consistent and universally reliable criterion has been found. The most widely accepted constraint is that the ``front velocity'' \cite{FoxKuper1970,Krotscheck1978}, or the ``asymptotic group velocity'' \cite{Pu2010}, both of which usually coincide by L'Hopital's rule, should not exceed the speed of light $c$ (=1, in our units), namely
\begin{equation}\label{vf}
v_f =\lim_{k \rightarrow \infty, \, k \in \mathbb{R}} \dfrac{\mathfrak{Re} \,\omega}{k} \, \stackrel{\text{L'H}}{=} \! \lim_{k \rightarrow \infty, \, k \in \mathbb{R}} \dfrac{d \, \mathfrak{Re} \,\omega}{dk} \in [-1,1] \, .
\end{equation}
Unfortunately, this condition is far from satisfactory as many famous acausal equations in physics fulfill \eqref{vf} even if the theory of partial differential equations tells us that they propagate information at infinite speeds \cite{CourantHilbert2_book,Rauch_book}. Three notable examples are the diffusion equation, the Euclidean wave equation, and the linearized Benjamin-Bona-Mahony (BBM) \cite{Benjamin1972} equation, respectively:
\begin{equation}\label{exeption}
\begin{split}
(\partial_t - \partial^2_x) \varphi={}& 0 \quad \Longrightarrow \quad  \omega=-ik^2 \, , \\
(\partial_t^2+\partial^2_x) \varphi={}& 0 \quad  \Longrightarrow \quad \omega= \pm i k \, , \\
(\partial_t+\partial_x -\partial_t \partial^2_x)\varphi={}& 0 \quad \Longrightarrow \quad \omega = \dfrac{k}{1+k^2} \, . \\
\end{split}
\end{equation}
All these equations have $v_f=0$. Furthermore, their phase and group velocities are (sub)luminal for all $k$. Nevertheless, these three models are strongly acausal. The BBM equation is particularly striking because one cannot attribute the causality violation to the imaginary part of $\omega$, given that $\omega$ is real for real $k$. Yet, it is acausal as the lines $t=\text{const}$ are spacelike characteristics\footnote{The superluminality of the BBM equation does not contradict \cite{Krotscheck1978}, see the Supplemental Material.}. 

This Letter shows that the limitations of \eqref{vf} are manifestations of a fundamental impossibility. Namely, unless $\omega(k)=a+bk$ for all $k$ (with $a,b$ constant), a single dispersion relation $\omega(k)$ cannot be causal. Rather, ``causality'' is a collective property of the system, which describes how all the excitation branches $\omega_n(k)$ \textit{combine} when the full initial value problem is set up. Therefore, apart from $\omega(k)=a{+}bk$, it is impossible to formulate a sufficient condition for causality in the form of an inequality that $\omega(k)$ should obey. This is why, given a causality criterion like \eqref{vf}, one can always find models that fulfill it and are acausal, such as  \eqref{exeption}. 

Nevertheless, we also show that one can overcome these difficulties by appealing to specific structures present in many (but not all) physical theories, which guarantee that the dispersion branches combine ``correctly'' to ensure causality. In particular, if the operator governing the dynamics is \emph{local} and the system is covariantly \emph{stable} (in precise senses defined below), \emph{all superluminal tails cancel out,} see Theorem \ref{theo1} for a precise statement.

\section{A key inequality}
Our analysis relies on the following inequality, which must hold for all dispersion branches describing disturbances around the equilibrium state of a stable system in relativity \cite{HellerBounds2022,GavassinoBounds2023}:
\begin{equation}\label{In3}
\mathfrak{Im} \, \omega(k) \leq |\mathfrak{Im} \, k| \, .
\end{equation}
This covariant bound can be derived from the study of retarded causal correlators of stable phases of matter, and it is textbook material \cite{Itzykson:1980rh}, whose importance in constraining transport properties of matter was demonstrated in Ref.\ \cite{HellerBounds2022}. Independently from the principle of causality, \eqref{In3} constitutes the physical requirement that a stable system should be simultaneously stable in every inertial frame of reference \cite{GavassinoSuperluminal2021}. In fact, if \eqref{In3} were violated, namely if there were some $k \in \mathbb{C}$ for which $\mathfrak{Im}\, \omega > |\mathfrak{Im} \, k|$, then a boost with velocity $v=\mathfrak{Im} k/\mathfrak{Im} \omega$ would lead us to a new reference frame where $ \mathfrak{Im} \, \omega' >0 $ and $\mathfrak{Im} \, k' =0$ \cite{GavassinoBounds2023}.
This would imply that there is an observer who can detect a growing Fourier mode, signaling an instability \cite{Hiscock_Insatibility_first_order,Kost2000,GavassinoLyapunov_2020}. For this reason, we assume \eqref{In3} holds as a basic stability property of the system.

\section{Single dispersion branches are superluminal} 
Fix some level of description of matter, which may be, e.g., quantum field theory, kinetic theory, or hydrodynamics. Using established techniques \cite{Birmingham2002,KovtunHolography2005,Kovtun2019,Perna2021}, one can compute all the (possibly infinite) dispersion branches predicted by such theory. Choose one of interest, $\omega(k)$. According to conventional wisdom \cite{Brillouin_book,DenicolKodamaMota2008,Koide2011,FoxKuper1970,Krotscheck1978,Pu2010}, the relation $\omega(k)$ determines how the corresponding excitation ``propagates'', and there should be some causality criterion for $\omega(k)$, e.g. \eqref{vf}, which guarantees that the excitation propagates subluminally. Now we prove that this intuitive interpretation can be consistently maintained only in the trivial case $\omega(k) =a+bk$. In dispersive media, causality can never be argued from $\omega(k)$ alone.

First, let us make the above (incorrect) intuition about the causality of $\omega(k)$ more precise.
Let $\varphi(x^\mu)\in \mathbb{C}$ be the linear perturbation to a local observable of interest. For example, $\varphi(x^\mu)$ may be the local energy density fluctuation. Then, consider a 1+1 dimensional profile $\varphi(t,x)$ that is constructed by superimposing plane waves all belonging \textit{solely} to the selected excitation branch $\omega(k)$, i.e.
\begin{equation}\label{gabuboi}
\varphi(t,x)= \int_{-\infty}^{+\infty} \varphi(k) e^{i[kx-\omega(k)t]} \dfrac{dk}{2\pi} \, .
\end{equation}
By setting $t=0$, we find that $\varphi(k)$ is the Fourier transform of the initial data, $\varphi(0,x)$.
The straightforward definition of ``causal dispersion relation'' is the following: If $\varphi(0,x)$ has support in a set $\mathcal{R}$, then the support of $\varphi(t,x)$ at later times should be contained inside the future lightcone of $\mathcal{R}$, see Figure \ref{fig:fig}. As a consequence, if $\varphi(0,x)$ has compact spatial support, one should find that $\varphi(t,x)$ has compact spatial support for each fixed $t>0$. Now we will show that, in practice, this is never the case for $\varphi$ given by \eqref{gabuboi}. On the contrary, single-branched excitations of the form \eqref{gabuboi} always ``travel'' at infinite speed unless $\omega{=}a{+}bk$ (i.e., when the medium is not dispersive).

% By setting $t=0$, we find that $\varphi(k)$ is the Fourier transform of the initial data, $\varphi(0,x)$.
% \textcolor{blue}{Causality implies that if $\varphi(0,x)$ has support on an interval $\mathcal{R}$, then the support of $\varphi(t,x)$ at later times should be contained inside the future lightcone of $\mathcal{R}$, see figure \ref{fig:fig}. }
% \textcolor{red}{This was stated as a natural definition of a ``causal dispersion relation.''
% But this is simply a definition of causality and not of ``causal dispersion relation" because the statement involves only $\varphi(0,x)$ and $\varphi(t,x)$, and never the dispersion relation. } As a consequence, if $\varphi(0,x)$ has compact spatial support, we should have that $\varphi(t,x)$ has compact spatial support for \textcolor{blue}{each fixed} $t>0$. \textcolor{red}{Just to avoid people interpreting for all t as meaning the whole future light cone, which is not compact.} Our goal now is to show that this is never the case \textcolor{blue}{for $\varphi$ given by \eqref{gabuboi}}, i.e., $\varphi(t,x)$ always ``travels'' outside the lightcone.

\begin{figure}
\begin{center}
\includegraphics[width=0.5\textwidth]{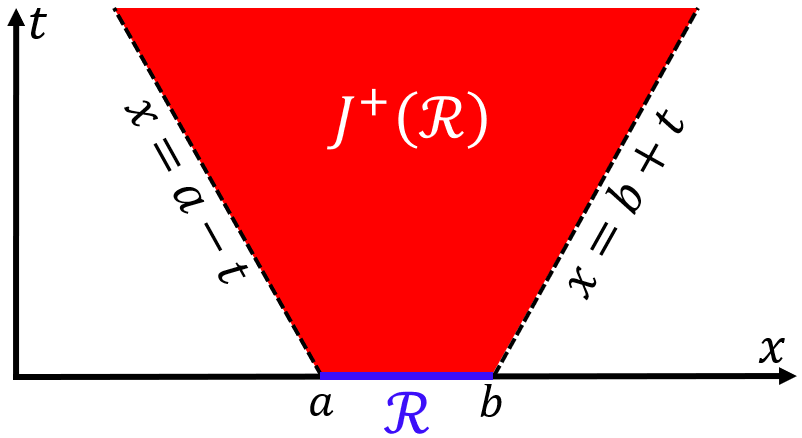}
	\caption{The principle of causality. If a perturbation has initial support inside a region of space $\mathcal{R}$ (blue segment), then it cannot propagate outside the set $J^+(\mathcal{R})$ called the ``causal future of $\mathcal{R}$'' \cite{Hawking1973}, or ``future lightcone of $\mathcal{R}$'' \cite{Susskind1969} (red region).}
	\label{fig:fig}
	\end{center}
\end{figure}

\subsection{A simple argument} If $\varphi(t,x)$ vanishes outside the future lightcone of a compact set $\mathcal{R}$, then also $\partial_t \varphi(t,x)$ must vanish there. Hence, to prove that $\varphi(t,x)$ exits the lightcone, it suffices to show that, for some $t_0>0$, the spatial profiles of $\varphi(t_0,x)$ and $\partial_t \varphi(t_0,x)$ cannot both have compact support simultaneously.
We assume that $\varphi(t,x)$ is smooth, but the argument can be generalized. 

Fix $t_0 >0$. From \eqref{gabuboi} and the uniqueness of the Fourier transform, we have that the spatial Fourier
transform of $\varphi(t_0,x)$
is given by  $\varphi(t_0,k) = \varphi(k) e^{-i\omega(k) t_0}$. Now, suppose that $\varphi(t_0,x)$ is compactly supported. Then, $\varphi(t_0,k)$ extends to an entire function of $k\in \mathbb{C}$ \cite{Hormander_book}. 
Under our assumptions, we can bring time derivatives under the integral to conclude that $\partial_t\varphi(t_0,x)$ has spatial Fourier transform $\dot{\varphi}(t_0,k)=-i\omega(k)\varphi(t_0,k)$. Corollary 1.1 of \cite{HellerBounds2022} tells us that, if $\omega(k)$ obeys \eqref{In3}, then it cannot be an entire function (unless $\omega {=} a{+}bk$). Therefore, $\dot{\varphi}(t_0,k)$ is the product of an entire function with a function that is not entire. Such a function can be entire only in the remote eventuality in which the discrete zeroes of $\varphi(k)$ happen to cancel the singularities of $\omega(k)$. 
This requires a perfect fine-tuning of the initial data, and it does not happen in general\footnote{Note that, if $\varphi(k)$ is the Fourier transform of a compactly supported function $\varphi(0,x)$, then $\varphi(k-a)$ is the Fourier transform of the function $\varphi(0,x)e^{iax}$, which is also compactly supported. Thus, if the zeroes of $\varphi(k)$ overlap the singularities of $\omega(k)$, we can always make a shift $a$ in Fourier space, and construct a new solution for which this overlap no longer happens.}. Furthermore, according to Theorem 2 of \cite{HellerBounds2022}, if $\omega(k)$ satisfies \eqref{In3}, then its singularities are never poles or essential singularities. Instead, they are expected to be branch points, which cannot be erased by multiplying $\omega(k)$ with an entire non-zero function.  Thus, $\dot{\varphi}(t_0,k)$ is not an entire function and, therefore, $\partial_t \varphi(t_0,x)$ cannot have compact support, as desired.

\subsection{Application 1: Hegerfeldt paradox} The above argument is a generalization of the well-known result (due to \citet{Hegerfeldt1974}) that relativistic single-particle wavefunctions of the form 
\begin{equation}\label{relatshrod}
\varphi(t,x)= \int_{-\infty}^{+\infty} \varphi(k) e^{i(kx-\sqrt{m^2{+}k^2} \,t)} \dfrac{dk}{2\pi} 
\end{equation} 
must propagate outside the lightcone \cite{Peskin_book}. Indeed, it can be easily verified that the dispersion relation of the free particle, $\omega=\sqrt{m^2{+}k^2}$, obeys \eqref{In3}, and this forces it to be non-analytical, as testified by the square root. Hence, the support of \eqref{relatshrod} expands at infinite speed \cite{Thaller1992}, even if the group velocity, $v_g(k)=k/\sqrt{m^2{+}k^2}$, is subluminal.

\subsection{Application 2: Necessity of non-hydrodynamic modes} An immediate corollary of our analysis is that the retarded Green's function of any theory for diffusion having only one dispersion relation, $\omega(k) = -iDk^2+\mathcal{O}(k^3)$, always exits the lightcone. Thus, to build a subluminal Green's function, we need at least two dispersion relations (see \cite{Morse1953} \S 7.4). This explains why an additional (usually gapped) mode is needed for causality \cite{HellerBounds2022}.

\subsection{Explanation} The superluminal behavior of \eqref{gabuboi} in causal matter seems absurd, but there is a simple explanation: excitations of the form \eqref{gabuboi} cannot be truly localized, unless $\omega {=}a{+}bk$. They may \textit{seem} to have compact support, if $\varphi(0,x)$ is supported in $\mathcal{R}$, but, in principle, an observer can detect the excitation from outside $\mathcal{R}$ already at $t=0$ by measuring some other observable. 
In fact, we recall that the dispersion relation $\omega(k)$ is derived from some underlying physical theory (e.g. quantum field theory, kinetic theory, or hydrodynamics), which may possess several other local observables besides $\varphi$. The fact that the initial profile $\varphi(0,x)$ has support inside $\mathcal{R}$ does not imply that all the measurable fields affected by the excitation are unperturbed outside $\mathcal{R}$. Instead, it may be the case that, due to this excitation, the perturbation to a second observable $\psi(x^\mu)$ of the theory has already unbounded support at $t=0$. This is how, in principle, $\varphi$ can propagate outside the lightcone without necessarily violating the principle of causality in the full theory: there is no superluminal propagation of information if such information was already accessible through the measurement of $\psi(0,x)$ outside $\mathcal{R}$. Indeed, below we prove that if the perturbations to \textit{all} the observables are initially supported inside a compact region $\mathcal{R}$ (i.e. the excitation is truly localized), and the dynamics is governed by a local operator, then $\varphi(t,x)$ cannot be expressed in the form \eqref{gabuboi}, and it must always combine at least two dispersion branches, unless $\omega=a{+}bk$.

\section{Compactly-supported excitations} We assume that the state of the system at a given time can be characterized, in the linear regime, by a collection of smooth perturbation fields $\Psi(x^\mu)\in \mathbb{C}^{\mathfrak{D}}$, which all vanish at equilibrium.
We assume that $\mathfrak{D}$ is finite, although it can be as large as the number of particles in a material volume element. In most physical theories currently available (e.g., electrodynamics, elasticity theory, or hydrodynamics), the 1+1 dimensional equation of motion of the system takes the form\footnote{Equations involving higher derivatives in time, e.g., $(\partial_t^2{-}\partial^2_x)\varphi{=}0$, can always be reduced to systems that are of first order in time through the introduction of more fields: $\partial_t \varphi = \Pi$, $\partial_t \Pi = \partial^2_x \varphi$.} 
%e.g. the Klein-Gordon equation $(\partial_t^2-\partial^2_x+m^2)\varphi=0$ can be decomposed into the system $\partial_t \varphi = \Pi$, $\partial_t \Pi = (\partial^2_x-m^2)\varphi$, with  degrees of freedom $\Psi=\{\varphi,\Pi \}$.}
\begin{equation}\label{belllo}
\partial_t \Psi = \mathcal{L}(\partial_x) \Psi \, ,
\end{equation}
where is $\mathcal{L}(\partial_x)$ a polynomial of finite degree in $\partial_x$, i.e. $\mathcal{L}(\partial_x)=A_0 +A_1 \partial_x+...+A_M \partial_x^M$, where $A_j$ are constant $\mathfrak{D}\times \mathfrak{D}$ matrices, and $M \in \mathbb{N}$. This is what we mean by a ``local operator''. In fact, operators involving an infinite series of derivatives can produce non-localities and causality violations. For example, if we set $\mathcal{L}(\partial_x)=e^{a\partial_x}$, equation \eqref{belllo} becomes $\partial_t \Psi(t,x)=\Psi(t,x{+}a)$, which is clearly a non-local theory. Indeed, the main reason why the BBM equation in \eqref{exeption} is acausal is that its dynamical operator, $\mathcal{L}(\partial_x)=(\partial_x^2-1)^{-1}\partial_x$, is non-local\footnote{The equation $i \, \partial_t \varphi= \sqrt{m^2-\partial^2_x}\, \varphi$ is also non-local for the same reason, and this is ultimately what gives rise to the Hegerfeldt paradox in relativistic quantum mechanics.} \cite{Congy2020}.

The general formal solution to \eqref{belllo} reads
\begin{equation}\label{DoFF}
\Psi(t,x) = \int_{-\infty}^{+\infty} e^{\mathcal{L}(ik)t}\Psi(k)e^{ikx} \dfrac{dk}{2\pi} \, ,
\end{equation}
where $\Psi(k)$ is the Fourier transform of the initial data $\Psi(0,x)$. Now, the field $\varphi(t,x)$, being a linearized local observable, is a local linear functional of the degrees of freedom, namely $\varphi=\mathcal{V}(\partial_x)\Psi$, where $\mathcal{V}$ is also a polynomial of finite degree in $\partial_x$, i.e. $\mathcal{V}(\partial_x)=B_0 +B_1 \partial_x+...+B_N \partial_x^N$. Here, $B_j$ are constant row vectors of length $\mathfrak{D}$, and $N \in \mathbb{N}$. Therefore, we have the following formula:
\begin{equation}\label{wood}
\varphi(t,x) = \int_{-\infty}^{+\infty} \mathcal{V}(ik) e^{\mathcal{L}(ik)t}\Psi(k)e^{ikx} \dfrac{dk}{2\pi} \, .
\end{equation} 
Assume that the excitation is initially supported inside $\mathcal{R}$. Then, all the components of $\Psi(0,x)$ are compactly supported, and all the components of $\Psi(k)$ are entire functions. Furthermore, $\mathcal{L}(ik)$ and $\mathcal{V}(ik)$ are entire in $k\in \mathbb{C}$, being polynomials. Also, the matrix exponential is an analytic function of the components of the matrix in the exponent, so that $e^{\mathcal{L}(ik)t}$ is entire in $k$. Combining these results, we can conclude that the integrand of \eqref{wood} is an entire function of $k$ for all $t$. For this reason, it cannot coincide with \eqref{gabuboi}, unless $\omega(k) =a+bk$ (i.e. in dispersion-free systems). This shows that, when we construct the state \eqref{gabuboi} in a dispersive medium, we implicitly allow some component of $\Psi$ to have unbounded support already at $t=0$, which is what we wanted to prove. 
In the Supplementary Material, we analyze the explicit example of the Klein-Gordon equation.

\section{Causality criterion for stable matter} The above analysis suggests that if the dynamics of the system is governed by a local operator $\mathcal{L}$, then all the dispersion branches will automatically combine in a way to cancel the infinite tails of the individual excitations \eqref{gabuboi}. This intuition can be made rigorous through the following theorem, according to which, schematically, 
$$
\binom{\text{Local}}{\text{equations}} + \binom{\text{Stability in}}{\text{all frames}} \quad \Longrightarrow \quad \binom{\text{Relativistic}}{\text{causality}} \, .
$$
More precisely:
\begin{theorem}\label{theo1}
If $\mathcal{L}(ik)$ and $\mathcal{V}(ik)$ are polynomials of finite degree in $ik$, and the eigenvalues $\omega_n(k)$ of $i\mathcal{L}(ik)$ obey the stability requirement \eqref{In3} for all $k \in \mathbb{C}$, then all smooth linear excitations propagate subluminally, in the sense that the support of $\varphi(t,x)$, as given by \eqref{wood}, is contained within the future lightcone of the support of $\Psi(0,x)$.
\end{theorem}
\begin{figure}
\begin{center}
\includegraphics[width=0.49\textwidth]{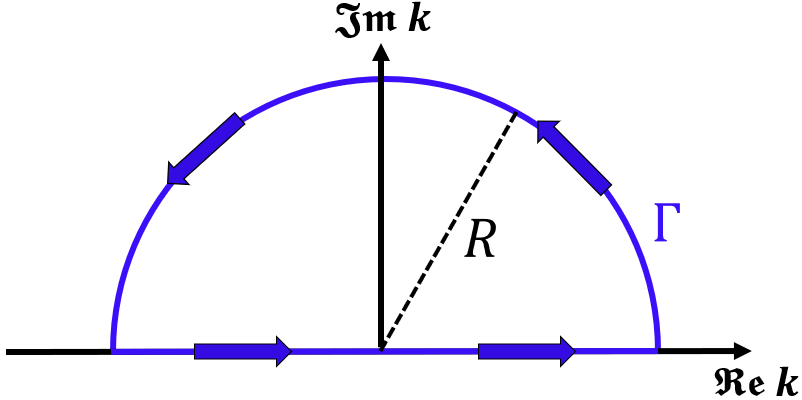}
	\caption{Path of integration %$\Gamma$ in the complex $k$-space 
for the proof of Theorem \ref{theo1}.}
	\label{fig:fig2}
	\end{center}
\end{figure}
\begin{proof}
We will focus on the case where $\Psi(0,x)$ has support inside the interval $[-1,0]$. We will verify that $\varphi(t,x)$ vanishes for $x>t \geq 0$. More general cases can be recovered from here by invoking linearity, translation invariance, and closure of the solution space.

Consider the complex integral
\begin{equation}\label{III}
\mathcal{I}=\int_\Gamma \mathcal{V}(ik) e^{\mathcal{L}(ik)t}\Psi(k)e^{ikx} \dfrac{dk}{2\pi} \, ,
\end{equation}
where $\Gamma$ is the closed loop in complex $k$ space in Fig.\ \ref{fig:fig2}, in the limit of large $R$. Since the integrand is entire, $\mathcal{I}=0$. Let us now show that, if $x>t \geq 0$, then the contribution coming from the upper semicircle decays to zero as $R\rightarrow +\infty$, so that $0=\mathcal{I}=\varphi(t,x)$, see equation \eqref{wood}. To this end, we first note that, according to the Jordan-Chevalley decomposition theorem, the matrix $\mathcal{L}(ik)$ can be expressed as 
\begin{equation}
\mathcal{L}(ik)=-i\sum_n \omega_n(k)\mathcal{P}_n(k)+\mathcal{N}(k) \, ,
\end{equation}
where $\mathcal{P}_n$ are complementary eigenprojectors (so that $\mathcal{P}_m\mathcal{P}_n=\delta_{mn}\mathcal{P}_n$, $\sum_n \mathcal{P}_n=\mathbb{I}$), and $\mathcal{N}$ is a nilpotent matrix ($\mathcal{N}^{a}{=}0$ for some $a \in \mathbb{N}$) which commutes with all $\mathcal{P}_n$. Thus, the integrand in \eqref{III} can be rewritten as
\begin{equation}\label{brutale}
\sum_n\sum_{j=0}^{a-1}\mathcal{V}  \dfrac{(\mathcal{N} t)^j}{j!} \mathcal{P}_n e^{i(kx-\omega_n t)} \Psi  \, .
\end{equation}
The matrix elements of $\mathcal{N}^j$ and $\mathcal{P}_n$ grow at most like powers of $|k|$. 
This follows from \cite[Chapter 2, eqs. (1.21) and (1.26)]{Kato_book}, applied to the matrix $(ik)^{-M} \mathcal{L}(ik)$ regarded as a polynomial in $(ik)^{-1} \! {\rightarrow} \, 0$, combined with the fact that $(ik)^{-M} \mathcal{L}(ik)$ and $\mathcal{L}(ik)$ have the same invariant subspaces. 
On the other hand, if $\mathfrak{Im} k \geq 0$, and $x>t \geq 0$, we have the following estimates:
\begin{equation}\label{tweve}
\begin{split}
& \big|e^{i(kx-\omega_nt)}\big|= e^{-x \, \mathfrak{Im}k+t \, \mathfrak{Im}\omega_n} \leq e^{-(x-t)\mathfrak{Im}k} \leq 1  \, , \\
& |\Psi(k)|=\bigg|\int_{-1}^0 e^{\tilde{x} \mathfrak{Im} k}\Psi(0,\tilde{x})e^{-i\tilde{x} \, \mathfrak{Re} k}d\tilde{x} \bigg| \leq L^1[\Psi(0,x)] \, . \\
\end{split}
\end{equation}
In the first line, we have invoked the inequality \eqref{In3}. In the second line, we have used the fact that $e^{\tilde{x} \mathfrak{Im} k} \leq 1$ inside the interval $[-1,0]$. Note that $\Psi(0,x)$, being continuous and compactly supported, has finite $L^1$ norm. From the estimates \eqref{tweve}, we can conclude that \eqref{brutale} decays exponentially to zero when $\mathfrak{Im} k \rightarrow +\infty$. Furthermore, since $\Psi(\mathfrak{Re} k +i \mathfrak{Im} k)$, regarded as a function of $\mathfrak{Re}k$, is the Fourier transform of the  Schwartz function $e^{x \mathfrak{Im} k}\Psi(0,x)$, it is itself a Schwartz function \cite{Hormander_book}, meaning that \eqref{brutale} decays to zero faster than any power also when $\mathfrak{Re}k \rightarrow \infty$. It follows that, as $R^2= (\mathfrak{Re}k)^2+(\mathfrak{Im}k)^2 \rightarrow +\infty$, the integral over the semicircle converges to zero since the integrand decays faster than any power of $R$.  
\end{proof}

Most derivations of \eqref{vf} rely on the assumption that $\omega \approx v_f k$ for large $k \in \mathbb{C}$, so that \eqref{vf} is a direct consequence of \eqref{In3}. However, \eqref{In3} is a much more stringent condition, as it automatically rules out the acausal equations \eqref{exeption}. Indeed, the apparent success of \eqref{vf} in many situations can  be traced back to \eqref{In3} through the following theorem proven below:
\begin{theorem}\label{theo2}
If \eqref{belllo} is a hyperbolic first-order system, with $\mathcal{L}(\partial_x)=-\Xi-\mathcal{M}\partial_x$, then \eqref{In3} implies \eqref{vf}, and the characteristic velocities coincide with the front velocities.
\end{theorem}
\begin{proof}
The ratios $\omega_n/k$ are eigenvalues of the matrix
\begin{equation}
\dfrac{i\mathcal{L}(ik)}{k} = \mathcal{M} + (ik)^{-1} \Xi \, .
\end{equation}
If we regard the right-hand side as a polynomial in $(ik)^{-1}$, we can take the limit as $(ik)^{-1}\rightarrow 0$ and apply the continuity property of eigenvalues %\cite[Chapter 2, $\S$1.2]
\cite{Kato_book} to conclude that $\omega_n/k$ must converge to eigenvalues of $\mathcal{M}$ for large $k \in \mathbb{C}$. But the eigenvalues of $\mathcal{M}$ are the characteristic velocities of the system, and they are real (by hyperbolicity), so that
\begin{equation}
\lim_{k \rightarrow \infty, \, k \in \mathbb{C} } \dfrac{\omega_n}{k} = v_{\text{ch},n}\in \mathbb{R} \, .
\end{equation}
Restricting the above limit to real $k$, and using the continuity of $\mathfrak{Re}$, we find that the characteristic velocities coincide with the front velocities. Restricting the limit to imaginary $k$, we find that \eqref{In3} implies causality. 
\begin{equation}
\begin{split}
& v_{\text{ch},n} = \mathfrak{Re} \bigg[ \lim_{k \rightarrow \infty, \, k \in \mathbb{R} } \dfrac{\omega_n}{k} \bigg] = \! \! \lim_{k \rightarrow \infty, \, k \in \mathbb{R} } \dfrac{\mathfrak{Re}\omega_n}{k} =v_{f,n} \, , \\
&  v_{\text{ch},n} = \mathfrak{Re} \bigg[ \lim_{k \rightarrow \infty, \, k \in i\mathbb{R} } \dfrac{\omega_n}{i \mathfrak{Im}k} \bigg] = \! \! \lim_{k \rightarrow \infty, \, k \in i\mathbb{R} } \dfrac{\mathfrak{Im}\omega_n}{\mathfrak{Im}k} {\in} [-1,1] \, . \\
  \end{split}
\end{equation} 
This completes our proof.
\end{proof}

While the above analysis was restricted to classical initial value problems, its broad implications may also be extrapolated to quantum systems. For example, some conformal field theories are known to be acausal \cite{Brigante:2008gz}. Given that such theories are local, we can ``apply'' our Theorem 1 to conclude that such theories are not covariantly stable and violate the bound \eqref{In3}, in agreement with Section III.A of \cite{GavassinoSuperluminal2021}.

\section{Correlators in QFT} Causality requires multiple dispersion relations also in QFT. Given a local observable operator $\hat{\varphi}(x^\mu)$, the correlator $G(x^\mu)=\braket{[\hat{\varphi}(x^\mu),\hat{\varphi}(0)]}$  has support inside the lightcone \cite{Peskin_book}. But since the slices of the lightcone at constant time are compact spheres, the spatial Fourier transform $G(t,\textbf{k})$ must be entire in $\textbf{k}$ for all $t$
%\cite[Theorem 7.1.14]
\cite{Hormander_book}. This is why introducing momentum cutoffs or ``patching'' correlators in momentum space leads to causality violations \cite{Henning_1995ft}: it breaks analyticity. Furthermore, if $G(t,\textbf{k})$ can be expressed as a superposition of modes of the form $e^{-i\omega_n(\textbf{k})t}$ (see e.g. \cite{Lowdon_2022yct}), then we know that all the non-analyticities of the individual frequencies $\omega_n(\textbf{k})$ must cancel out.

\section{Final remarks}Consider the following puzzle: All solutions of the relativistic Schr\"{o}dinger equation  $i\partial_t \varphi {=} \sqrt{m^2{-}\partial^2_x}\, \varphi$
are also solutions of the Klein-Gordon equation $-\partial^2_t\varphi {=} (m^2{-}\partial^2_x)\varphi$. Nevertheless, the former is notoriously acausal \cite{Peskin_book}, while the latter is causal. This defies the intuition of causality as a statement about the propagation speed of $\varphi$. How can the same function $\varphi(t,x)$ be superluminal when viewed as a solution of one equation and subluminal when viewed as a solution of another equation?

Here, we solved this puzzle by showing that causality is not an intrinsic property of the fields themselves. Rather, it is a property of how we ``attach information'' to the fields by defining the physical state. The existence of faster-than-light motion does not result in causality violation if the motion carries no new information about the state. Indeed, relativistic Schr\"odinger and Klein-Gordon differ by the way they define the physical state at a given time: $\{\varphi(x)\}$ in the former, and $\{\varphi(x),\partial_t \varphi(x)\}$ in the latter. The puzzle arises because compactly supported field states within relativistic Schr\"odinger (i.e., localized $\varphi$ profiles) must have unbounded support within Klein-Gordon (i.e., cannot be localized in $\partial_t \varphi$), see Supplementary Material.

Starting from this intuition, we showed that non-hydrodynamic modes become necessary for relativistic viscous hydrodynamics for the same reason that antiparticles are necessary for relativistic quantum mechanics: defining a notion of locality in dispersive systems requires at least two dispersion relations.

\textit{Note Added. -} Recently, other formulations of Theorems 1 and 2 were proposed \cite{WangPu2023,HoultKovtunCausality2023}.

\section*{Acknowledgements}
L. G. and J.N. would like to thank P. Kovtun and R. Hoult for a stimulating discussion. We also thank P. Lowdon for bringing the issue of correlators in QFT to our attention. L. G. is partially supported by a Vanderbilt Seeding Success Grant. M.M.D. is partially supported by NSF Grant No. DMS-2107701, a Chancellor’s Faculty Fellowship, DOE Grant No. DE-SC0024711, and a Vanderbilt Seeding Success Grant. J. N. is partially supported by the U.S. Department of Energy, Office of Science, Office for Nuclear Physics under Award No. DE-SC0023861. L.G. and J.N. would like to thank P. Kovtun and R. Hoult for a stimulating discussion. We also thank P. Lowdon for bringing the issue of correlators in QFT to our attention. The authors thank KITP Santa Barbara for its hospitality during “The Many Faces of Relativistic Fluid Dynamics” Program. This research was supported in part by the National Science Foundation under Grant No. NSF PHY-1748958.

\bibliography{Biblio}

\newpage

\onecolumngrid
\newpage
\begin{center}
  \textbf{\large Dispersion relations alone cannot guarantee causality\\Supplementary Material}\\[.2cm]
  L. Gavassino,$^{1}$ M. Disconzi,$^{1}$ and J. Noronha$^2$\\[.1cm]
  {\itshape ${}^1$Department of Mathematics, Vanderbilt University, Nashville, TN, USA\\
  ${}^2$Illinois Center for Advanced Studies of the Universe \& Department of Physics,\\University of Illinois at Urbana-Champaign, Urbana, IL 61801-3003, USA\\}
(Dated: \today)\\[1cm]
\end{center}

\setcounter{equation}{0}
\setcounter{figure}{0}
\setcounter{table}{0}
\setcounter{page}{1}
\renewcommand{\theequation}{S\arabic{equation}}
\renewcommand{\thefigure}{S\arabic{figure}}

\section{The BBM equation and front velocity theorems}

The superluminal behavior of the BBM dispersion relation
\begin{equation}
\omega= \dfrac{k}{1+k^2} 
\end{equation}
seems to contradict Proposition 2 of \cite{Krotscheck1978}, which gives sufficient conditions for causality if the front velocity $v_f$ is subluminal. However, there is no contradiction, since the BBM equation breaks assumption E, namely that the initial value problem for the excitations should be uniquely solvable. BBM \cite{Benjamin1972} have proved that the solutions to the differential equation $(\partial_t + \partial_x - \partial_t \partial^2_x)\varphi=0$ are unique, but this is true only within a space of functions which decay at spacelike infinity, and not for generic smooth functions. Indeed, since the lines $t=\text{const}$ are characteristics of the BBM equation, we are actually dealing with a characteristic initial value problem (not a proper Cauchy problem), and the Holmgren uniqueness theorem no longer applies \cite{Rauch_book}. This allows for the formation of non-zero solutions with vanishing initial profile $\varphi(0,x)$, even if the equation is of first-order in time. Such solutions cannot be expressed in the form
\begin{equation}
\varphi(t,x)= \int_{\mathbb{R}}\varphi(k)e^{i(kx-\omega t)}\dfrac{dk}{2\pi}
\end{equation}
because the Fourier amplitudes $\varphi(k)=\int_{\mathbb{R}} \varphi(0,x)e^{-ikx}dx$ would vanish. Instead, these solutions are unbounded in space, and they exist only if we allow for complex $k$. In particular, given that $\omega=k/(1+k^2)$ has two poles, at $k=\pm i$, we can construct the integral
\begin{equation}\label{fourz}
\varphi(t,x)= \int_\Gamma e^{i(kx-\omega t)} \dfrac{dk}{2\pi} \, ,
\end{equation}
where $\Gamma$ is a closed loop around one pole, say, $i$. At $t=0$, the integrand is $e^{ikx}$, which is entire, and thus $\varphi(0,x)=0$. However, the time derivative of the integrand has a pole, so that
\begin{equation}\label{ferguz}
\partial_t \varphi(0,x) =  \int_\Gamma  \dfrac{k e^{ikx}}{1+k^2} \dfrac{dk}{2\pi i} = \dfrac{e^{-x}}{2} \neq 0 \, .
\end{equation} 
It is evident that these solutions diverge at spatial infinity, so that they can be ruled out by requiring boundedness, as done by BBM \cite{Benjamin1972}. But the ``uniqueness property'' invoked by \citet{Krotscheck1978} is a much more stringent requirement. This becomes apparent in Proposition 3 of \cite{Krotscheck1978}, where it is explicitly shown that assumption E forbids the existence of poles in $\omega(k)$, thereby ruling out the BBM equation.

\newpage

\section{Stability, causality, and localized states of the Klein-Gordon equation}

\subsection{Covariant stability}

First of all, let us show explicitly that the dispersion relations $\omega{=}\pm \sqrt{m^2{+}k^2}$ obey the stability bound $\mathfrak{Im} \, \omega \leq |\mathfrak{Im} \, k|$. To this end, let us define the four-vector $p^\mu =(\omega,k)$, so that the dispersion relations imply $p_\mu p^\mu=-m^2$, where we adopt the metric signature $(-,+,+,+)$. Splitting the real and imaginary parts, we obtain
\begin{equation}
\begin{split}
& \mathfrak{Im}p_\mu \mathfrak{Im}p^\mu =m^2+\mathfrak{Re}p_\mu \mathfrak{Re}p^\mu \, ,  \\
& \mathfrak{Im}p_\mu \mathfrak{Re}p^\mu =0 \, . \\
\end{split}
\end{equation}
Suppose that $\mathfrak{Im}p^\mu$ is timelike. Then, from the second equation, we find that $\mathfrak{Re}p^\mu$ is spacelike. But then $\mathfrak{Re}p_\mu \mathfrak{Re}p^\mu \geq 0$, and the first equation implies $\mathfrak{Im}p_\mu \mathfrak{Im}p^\mu \geq m^2$, contradicting our original assumption that $\mathfrak{Im}p^\mu$ was timelike. Therefore, $\mathfrak{Im}p^\mu$ cannot lay inside the lighcone, so that necessarily $\mathfrak{Im} \, \omega  \leq |\mathfrak{Im} \, k|$.

\subsection{Locality and causality}

The Klein-Gordon equation $(\partial_t^2-\partial^2_x+m^2)\phi=0$ is of second order in time, and it can be decomposed into a system of two equations that are of first order in time:
\begin{equation}\label{acb}
\begin{split}
\partial_t \phi ={}& \Pi \, ,\\
\partial_t \Pi ={}& (\partial^2_x-m^2)\phi \, . \\
\end{split}
\end{equation}
Hence, the degrees of freedom are $\Psi=(\phi,\Pi)^T$, and the dynamics of the system has the form $\partial_t \Psi = \mathcal{L}(\partial_x)\Psi$, with
\begin{equation}\label{gringone}
\mathcal{L}(ik)= 
\begin{bmatrix}
0 & \, \, 1 \\
-(m^2{+}k^2) & \, \, 0 \\
\end{bmatrix} \, .
\end{equation}
Clearly, $\mathcal{L}$ is a local operator, being a second-order polynomial in the derivatives. Thus, our Theorem 1 applies, and the Klein-Gordon equation is causal. The general solution to  \eqref{acb} reads
\begin{equation}\label{grump}
\begin{bmatrix}
\phi(t,x)\\
\Pi(t,x) \\
\end{bmatrix} =\int_{-\infty}^{+\infty}
\begin{bmatrix}
\cos (\omega t) & \, \, \omega^{-1} \sin(\omega t) \\
-\omega \sin(\omega t)  & \cos (\omega t)  \\
\end{bmatrix}
\begin{bmatrix}
\phi(k) \\
\Pi(k) \\
\end{bmatrix} e^{ikx} \dfrac{dk}{2\pi} \, ,
\end{equation}
with $\omega=\sqrt{m^2+k^2}$. Despite the square root, the matrix in \eqref{grump} is entire in $k$, being the exponential of \eqref{gringone}. This is reflected in the fact that the following functions are entire:
\begin{equation}
\begin{split}
\cos \sqrt{z} ={}& \sum_{n=0}^{+\infty} \dfrac{(-1)^n z^n}{(2n)!} \, , \\
\dfrac{\sin\sqrt{z}}{\sqrt{z}}={}& \sum_{n=0}^{+\infty} \dfrac{(-1)^n z^n}{(2n+1)!} \, . \\
\end{split}
\end{equation}

\subsection{Positive-frequency solutions cannot be localized}

If both degrees of freedom $\phi(0,x)$ and $\Pi(0,x)$ have support inside a compact set $\mathcal{R}$, then both $\phi(k)$ and $\Pi(k)$ are entire functions of $k \in \mathbb{C}$ \cite{Hormander_book}. On the other hand, the first component of equation \eqref{grump} can be decomposed as follows:
\begin{equation}
    \phi(t,x) = \int_{-\infty}^{+\infty}
\bigg[\bigg(\phi(k) + \dfrac{\Pi(k)}{i\omega} \bigg)e^{i\omega t}+\bigg(\phi(k) - \dfrac{\Pi(k)}{i\omega} \bigg)e^{-i\omega t} \bigg] e^{ikx} \dfrac{dk}{4\pi} \, .
\end{equation}
As we can see, the only way for $\phi(t,x)$ to be a superposition involving only the ``positive-frequency waves'', $e^{-i\omega t}$, is to impose $\Pi(k)=-i\omega \phi(k)$. But this contradicts the assumption that both $\Pi(k)$ and $\phi(k)$ are entire, since $\omega = \sqrt{m^2+k^2}$ has two branch cuts. It follows that compactly-supported solutions are necessarily superpositions of both a positive and a negative frequency part. This superposition results in a perfect cancellation of the superluminal tails of the two individual parts, and this saves causality.

\newpage

%\bibliography{Biblio}

\label{lastpage}

\end{document}